\documentclass[letterpaper,10pt]{article}

\usepackage[svgnames]{xcolor}
\usepackage[numbers,longnamesfirst]{natbib}
\PassOptionsToPackage{backref=page}{hyperref}

\usepackage{xstring}
\usepackage{amsmath}
\usepackage{amssymb}
\usepackage{amsthm}
\usepackage{thm-restate}
\usepackage{mathabx}
\usepackage{mathtools}
\usepackage{todonotes}
\usepackage{enumitem}
\usepackage[ruled,boxed,noend]{algorithm2e}
\usepackage{hyperref}

\usepackage{microtype}
\usepackage{libertine}
\usepackage{libertinust1math}

\usepackage{tikz}

\def\titl{Random primes without primality testing}
\def\auts{Pascal Giorgi, Bruno Grenet, Armelle Perret du Cray, Daniel S. Roche}

\hypersetup{
  pdfauthor={\auts},
  pdftitle={\titl},
  colorlinks,
  linkcolor=DarkBlue,
  citecolor=DarkGreen,
  urlcolor=DarkBlue,
  bookmarks,
  bookmarksnumbered,
  pdfpagelabels=true,
}

\usepackage[nameinlink,capitalise]{cleveref}

\SetAlgoVlined
\DontPrintSemicolon
\SetKw{break}{break}
\SetKwInOut{Input}{Input}
\SetKwInOut{Output}{Output}
\SetKwFor{Loop}{repeat}{}{}

\usepackage{underscore}
\newcommand{\doi}[1]{\textsc{doi}: \href{http://dx.doi.org/#1}{#1}}

\newtheorem{theorem}{Theorem}[section]
\newtheorem{fact}[theorem]{Fact}
\newtheorem{lemma}[theorem]{Lemma}
\newtheorem{corollary}[theorem]{Corollary}
\newtheorem{definition}[theorem]{Definition}

\numberwithin{equation}{section}

\newcommand{\ZZ}{\ensuremath{\mathbb{Z}}}
\newcommand{\NN}{\ensuremath{\mathbb{N}}}
\newcommand{\QQ}{\ensuremath{\mathbb{Q}}}
\newcommand{\RR}{\ensuremath{\mathbb{R}}}

\newcommand{\R}{\ensuremath{\mathsf{R}}}
\newcommand{\GF}[1]{\ensuremath{\mathbb{F}_{#1}}}

\newcommand{\bnd}[2]{\ensuremath{#1\mathopen{}\left(#2\right)\mathclose{}}}
\newcommand{\oh}[1]{\bnd{O}{#1}}
\newcommand{\softoh}[1]{\bnd{\widetilde{O}}{#1}}
\newcommand{\floor}[1]{\ensuremath{\left\lfloor#1\right\rfloor}}

\DeclareMathOperator{\llog}{loglog}

\newcommand{\NewModulus}{\textsc{NewModulus}}

\newcommand{\RandUpdate}{\textsc{RandUpdate}}
\newcommand{\RandMod}{\textsc{RandMod}}

\newcommand{\inx}{\ensuremath{\mathbf{x}}}
\newcommand{\origalg}{\ensuremath{\mathcal{A}}}
\newcommand{\transalg}{\ensuremath{\overline{\origalg}}}
\newcommand{\bounds}{\ensuremath{\mathcal{B}}}
\renewcommand{\backref}[1]{Referenced on %
  \expandarg\StrCount{#1}{,}[\ncommas]%
  \ifthenelse{\ncommas = 0}{page~#1}%
  {pages~\StrBefore[\ncommas]{#1}{,}\ and\StrBehind[\ncommas]{#1}{,}}%
.}

\newcommand\auth[4]{%
  \begin{minipage}{.45\textwidth}%
  \centering
      #1\\%
      \normalsize
      #2\\%
      #3\\%
      #4
  \end{minipage}%
}
\title{\titl}
\author{%
\auth{Pascal Giorgi}{LIRMM, Univ. Montpellier, CNRS}{Montpellier, France}{pascal.giorgi@lirmm.fr}
\auth{Bruno Grenet}{LIRMM, Univ. Montpellier, CNRS}{Montpellier, France}{bruno.grenet@lirmm.fr}
\\[3em]
\auth{Armelle Perret du Cray}{LIRMM, Univ. Montpellier, CNRS}{Montpellier, France}{armelle.perret-du-cray@lirmm.fr}
\auth{Daniel S. Roche}{United States Naval Academy}{Annapolis, Maryland, U.S.A}{roche@usna.edu}
}

\begin{document}
\maketitle
\begin{abstract}
  Numerous algorithms call for
  computation over the integers modulo a ran\-domly-chosen large prime.
  In some cases, the quasi-cubic complexity of selecting a
  random prime can dominate the total running time. We propose a new
  variant of the classic D5 algorithm for ``dynamic evaluation'', applied
  to a randomly-chosen (composite) integer. Unlike the D5 principle
  which has been used in the past to compute over a direct product of
  fields, our method is simpler as it only requires following a single
  path after any modulus splits. The transformation we propose can apply
  to any algorithm in the algebraic RAM model, even allowing
  randomization. The resulting transformed algorithm avoids any
  primality tests and will, with constant positive probability, have the same
  result as the original computation modulo a randomly-chosen prime. As
  an application, we demonstrate how to compute the exact number of
  nonzero terms in an unknown integer polynomial in quasi-linear time.
  We also show how the same algorithmic transformation technique can be
  used for computing modulo random irreducible polynomials over a finite
  field.
\end{abstract}
                   
\section{Introduction}

Consider the following situation which arises commonly in exact
computational problems: We have a problem over the integers $\ZZ$ to
solve, but perhaps due to intermediate expression swell or the need for
exact divisions, solving direcly in $\ZZ$ or $\QQ$ is infeasible. There
may be fast algorithms for this problem over a finite field, say modulo a
prime $p$, but then some (unknown) conditions on the prime $p$ must be
met in order for the solution over $\GF{p}$ to coincide with the actual
integer solution. Typically, $p$ must not be a divisor of some unknown
(but bounded) large value.

Then a classical approach, called the \emph{big-prime technique}, is to
randomly choose a large prime $p$, with sufficient bit-length to
overcome any non-divisibility conditions with high probability, solve
the problem over $\GF{p}$, and return the result; see
\citep[Chapter 5]{Gathen:2013}.

When the computation itself is expensive, we can
ignore the cost of prime number generation for practical purposes; the
running time to compute $p$ will be dwarfed by the computations within
$\GF{p}$ that follow.
But as the computation becomes more efficient, and particularly for
algorithms that have quasi-linear complexity, the cost of prime number
generation is more significant, and may be a performance bottleneck in
theory and/or in practice.

The most efficient Monte Carlo method to
generate a random prime number is to sample random integers of the
required size and then test them for primality. To find a $b$-bit prime,
using fast arithmetic, each primality test costs $\softoh{b^2}$ time,
and because the density of primes among all $b$-bit numbers is
proportional to $1/b$, the total cost of prime generation in this way is
$\softoh{b^3}$. Faster practical techniques by \citep{Mau95,JPV00}
incorporate many clever ideas but do not improve on this cubic
complexity bound; see \citep[Chapter 10]{Shoup:2008}.

In summary: the cubic cost of generating a large probable-prime may
dominate the total cost of the big-prime method when the needed primes
are quite large and the mod-$p$ algorithm is quite fast. The aim of this
paper is to tackle this issue with a general technique that uses
random moduli, \emph{not necessarily prime}, with various extra checks
along the way, to provably get the same result in many cases as would be
achieved by explicitly generating a random prime.

An analogous situation occurs for polynomials over finite fields: if the
original problem is in a fixed finite field $\GF{q}$, then in many cases
one needs to compute over an extension field $\GF{q^k}$ for $k$
sufficiently large. This then poses the challenge of computing a random
irreducible degree-$k$ polynomial in $\GF{q}[x]$, which again takes at
least quasi-cubic time using the best current methods. We will show that
our same basic algorithm transformation technique for integers applies
in this case as well, with similar probability bounds.

\subsection{Algorithmic transformation technique}

Our method builds on the long
line of techniques known as \emph{dynamic evaluation}, or the D$^5$
principle
\citep{DDDD85}. This is a very
general technique which has since been employed for a wide range of
computational problems
\cite{duval89,lazard91,DR94,mr95,DMSX06,Noro06,DL10,HL20,Poteaux21}.
For our purposes the idea is to start
computing modulo a possibly-composite $m$, and ``split'' the evaluation
with an efficient GCD computation whenever we need to test for zero or
perform a division.

Like the recent \emph{directed} evaluation technique proposed by
\citet{HL20}, we opt to first take the larger-size branch in any GCD
splitting. But unlike their method (and prior work) which is ultimately
focused on recovering the correct result over the original product of
fields, here the goal is only to have an answer which is
consistent with what it \emph{would have been} in some randomly-chosen
finite field. For that reason, we can ignore the smaller branch of any
split, avoiding the reconstruction process altogether.

In terms of running time, the key observation is that GCDs can be
computed in quasi-linear bit complexity using the half-GCD algorithm of
\citep{Moenck73,BGY80}, and therefore this transformation has the same
soft-oh bit complexity as it would to compute over an actual prime of
the desired size. We save because there is no longer a need to actually
generate (and test for) the random prime.

Proving the probabilistic correctness of this approach is the main
challenge and contribution of our paper. First we need good estimates on
the probability that a random integer has a large prime factor. In our
technique, having a single prime factor $p\divides m$ where $p^2\ge m$
is necessary and sufficient.

Second, and more challengingly, we need to allow for a general model of
computation where the algorithm may sample uniformly from the random
field it is computing over, and probabilistic correctness over a
randomly-chosen field $\GF{p}$ should carry over to probabilistic
correctness of our method with random not-necessarily-prime moduli.

To our knowledge, previous applications of the D$^5$ principle have
considered only computational models which are deterministic and do not
allow random sampling of field elements. \citet{HL20} suggest that this
can be overcome by providing the deterministic computation tree with a
``pool'' of pre-selected random field elements, and the recent paper of
\citet{NSSV21} mentions this limitation and also skirts around it by
providing pre-selected ``random'' choices as additional algorithm
inputs.

But this pre-selection does not really make sense in our setting, where
the initial modulus $m$ itself is randomly chosen as well, as we must do
in order to avoid ``unlucky'' choices of the underlying finite field $\GF{p}$.
Which is to say, for any \emph{fixed}, pre-selected value of $m$ and
``random'' elements modulo $m$, there is no way to argue that the result
will be correct with high probability; but it is also impossible to
choose uniform-random elements modulo $m$ or $p$ without knowing the
modulus in advance. Instead, we take care to actually allow
randomization within the transformed algorithm, and prove that
probabilistic correctness modulo most sufficiently large primes $p$
does indeed imply probabilistic correctness modulo a large-enough random
integer $m$ (and equivalently with random extension fields over a fixed
finite field).

Our optimization to compute on only one branch in fact adds new wrinkles
to the challenge of proving correctness. As a small illustration,
consider a very simple algebraic algorithm which simply chooses a random
field element and tests whether it is zero. Over $\GF{2}$ there should
clearly be a $\frac{1}{2}$ probability of each outcome. But if we
instead compute with initial modulus $m=30$ and use our ``largest
branch'' splitting technique, whenever $2$ divides the final modulus,
there is only a $\frac{1}{3}$ chance of getting zero over $\GF{2}$.
The reason this
can occur is that the branches, and hence the choice of which
modulus to use at the end, \emph{may themselves depend on previous
random choices}. In fact it is not hard to construct pathological
algorithmic examples which are usually correct modulo some certain-sized
primes $p$, but usually \emph{incorrect} modulo $m$ when $m$ is a multiple of
$p$. Overcoming such issues in a proven and generic way is a major
challenge of the present investigation.

\subsection{Application to sparsity determination}

As an important application of our technique to avoid primality testing
in randomized computation, and indeed our original motivation for this
work, we develop a new algorithm to compute the sparsity of an unknown
black-box polynomial.

This Monte Carlo randomized
algorithm uses samples of a \emph{modular black box}, via which an unknown
sparse integer polynomial $f\in\ZZ[x_1,\ldots,x_n]$ can be evaluated for
any chosen modulus $m\in\NN$ and point
$(\theta_1,\ldots,\theta_n)\in[0,m)^n$, as well as bounds $H$ and $D$
on the height and max degree, respectively, to determine the number of
nonzero terms $\#f$ in the unknown polynomial, correct with high
probability. The bit complexity (accounting for black box evaluations)
is softly-linear in the size $f$; see \cref{thm:getsparsity-z} for a precise
statement.

This problem is closely related to the more general problem of sparse interpolation
\cite{BenorTiwari:1988,KaltofenLakshman:1988,KaltofenTrager1990,gs09,JM10,Kal10a,AGR15,HL15,arn16,Huang:2019},
where all coefficients and exponents of $f$ are to be recovered.
Oftentimes such algorithms assume they are given an upper bound
$T\ge\#f$ and have running time proportional to this $T$, so having a
fast way to determine $\#f$ exactly can be valuable in practice.

Our method mostly follows the \emph{early termination} strategy by
\citet{KL03}, which was the first efficient sparse interpolation algorithm
not to require an \emph{a priori} upper bound $T\ge\#f$. We save on the
running time by explicitly stopping the algorithm early as soon as $\#f$
is learned, and avoiding costly later steps which require special field
structure.

Moreover, while the algorithm of \citep{KL03} works over a more general
domain, its bit complexity over $\ZZ$ is exponentially large; thus, a
``big prime'' technique is commonly employed; see
\cite{kaltofen90,Kal10a,JM10}.
But now that we have
reduced the arithmetic complexity to quasi-linear, the cubic bit-cost of
large prime generation becomes significant. This is where our new
algorithm transformation technique comes into play: we use our new
methods to obtain the same results as if working modulo a random prime,
while actually computing modulo a random composite number with at most
twice the bit length.

\subsection{Summary of contributions}

The main results of this paper are:
\begin{itemize}
  \item A new variant of the D$^5$ principle which focuses on computing
    modulo random primes rather than in a given product of fields
    (\cref{ssec:trans});
  \item A careful analysis which shows that any algorithm which is
    \emph{probably} correct for \emph{most} random, sufficiently-large
    primes, can be solved without the cost of prime number generation
    using our new technique (\cref{ssec:analysis});
  \item A new algorithm with nearly-optimal bit complexity to determine
    the number of nonzero terms of an unknown sparse integer polynomial
    (\cref{sec:sparsity}); and
  \item An adaptation of the same techniques to computing modulo random
  irreducible polynomials over finite fields without irreducibility
  testing (\cref{sec:poly}).
\end{itemize}

\section{Prime density and counting bounds}\label{sec:primes}

In this section we review some mostly-known results on the number of
primes in intervals with certain properties, that will be needed for the
probabilistic analysis of our main results.

Throughout, we use the notation $[a,b)$ to denote the set of integers
$n$ satisfying $a\le n<b$, and we use the term \emph{$b$-bit integer} to mean an integer in
the range $[2^{b-1},2^b)$. Note that there are $2^{b-1}$ integers with
bit-length $b$.

\begin{definition}
  For any positive integers $m$ and $b$, we say $m$ is $b$-fat if
  $m$ has a prime divisor $p \ge 2^b$.
\end{definition}

This follows the definition of $M$-fat in the classic paper of
\citet{KR87}, except that we focus only on power-of-two bounds.

We first prove that at least half of all $2b$-bit integers have a prime factor with at
least $b$ bits, largely
following \citep[Lemma 8]{KR87}, which in turn is based on bounds from
\citet{RS62}.

\begin{lemma}\label{lem:bfat}
  For any $b \ge 1$, the number of $(2b)$-bit integers which are $b$-fat
  is at least $2^{2b-2}$.
\end{lemma}
\begin{proof}
  The claim is verified numerically for $1 \le b \le 6$.

  If $b\ge 7$, then
  Lemma 8 of \citep{KR87} tells us that the number of $b$-fat integers
  in the range $[1,2^{2b}]$ is at least $2^{2b-1}$.

  We need to show that at least half of the $b$-fat integers are in top
  half of this range.

  For any prime $p$ with $p\ge 2^b$, consider the multiples of $p$ in
  the range $[1,2^{2b}]$. By definition, all such multiples are $b$-fat.
  Let $k\in\NN$ so that $kp$ is the
  largest multiple of $p$ less than $2^{2b-1}$. Thus exactly $k$
  multiples of $p$ have bit-length strictly less than $2b$.
  And because $2kp < 2^{2b}$, there are at least $k$ multiples of $p$
  with exactly $2b$ bits. Incorporating the fact that $2^{2b}$ is never
  a multiple of $p$, we see that at least half of the multiples of $p$
  in the range $[1,2^{2b}]$ have bit-length exactly $2b$.

  Because any number less than $2^{2b}$ can only be divisible by at most
  one prime $p\ge 2^b$, the sets of multiples for $2^b \le p < 2^{2b}$
  in fact form a partition of the $b$-fat numbers with at most $2b$
  bits. Therefore, summing over all sets in this partition, we see that
  the total number of $b$-fat integers with $2b$ bits is at least
  $2^{2b-2}$.
\end{proof}

Many algorithms which perform computations modulo a random prime in fact
need to avoid a certain number of unlucky or ``bad'' primes. Here we
give an upper bound on the chance that a $b$-fat number is divisible by
a large bad prime. The proof is trivially just dividing the range by
$p$.

\begin{lemma}\label{lem:pdiv-prob}
  For any $b\ge 1$ and prime $p\ge 2^b$, at most $2^{b-1}$ integers with
  bit-length $2b$ are multiples of $p$.
\end{lemma}

\section{Computational model}\label{sec:model}

Our main result is a transformation which, roughly speaking, takes any
algorithm in the algebraic RAM model for any chosen prime $p$,
and converts it into a randomized algorithm in the
(non-algebraic) RAM model that produces the same output most of the time
while avoiding prime number generation.

Here we must take care to define the requirements on the initial
algebraic algorithm. Prior work such as
\citep{DDDD85,DMSX06,HL20}
considered more general settings beyond computing in random finite
fields, but in fairly restricted models of deterministic algebraic
computation such as straight-line programs or computational trees.
Here we have a more restricted algebraic setting but a more general
computational one, allowing for loops, memory, and pseudorandom integer or
field element generation.

\subsection{Modular PRNGs}\label{ssec:prng}

We define a \emph{modular pseudo-random number generator}, or modular
PRNG, as a pair of deterministic algorithms:

\begin{itemize}
  \item $\RandMod(s,m) \to x$ takes a fixed-length state $s$ and any
  positive integer $m$ and produces
  a pseudorandom value $x$ uniformly distributed from the range $[0,m)$.
  \item $\RandUpdate(s) \to s'$ takes the current state and produces a
  value (with the same bit-length) for the next state.
\end{itemize}

(In practice there would also be an initialization procedure which takes
a smaller seed value to produce the initial state, but this detail is
unimportant for our discussion. That is, we treat the seed as synonymous
with the initial state $s$.)

Conceptually, both of these functions should be indistinguishable from
random.
More precisely, define $r(s,i,m)$ as the $i$'th output modulo $m$ from initial
state $s$, that is,

\(\RandMod(\underbrace{\RandUpdate(\cdots(\RandUpdate}_{i\text{ times}}(s))), m).\)

Then, even with an oracle to compute $r(s,i',m')$ for any tuple (i',m')
not both equal to $(i,m)$, it should be infeasible to distinguish
$r(s,i,m)$ from a truly random output modulo $m$.

\subsection{Algebraic RAM with integer I/O}\label{ssec:aram}

Our model of computation is a classical \emph{random access machine}
(RAM), with an ``integer side'' and an ``arithmetic side''. The integer
side is a normal RAM machine with instructions involving integer inputs
and outputs, and the arithmetic side involves \emph{only} the following
computations with elements of an arbitrary field:

\begin{itemize}
  \item Ring arithmetic: $+$, $-$, $\times$
  \item Multiplicative inverse (which in turn allows exact division)
  \item Pseudo-random generation of a field element
\end{itemize}

The multiplicative inverse computation may fail, e.g., on division by
zero, in which case the algorithm returns the special symbol $\bot$.

There are also two special instructions which take input from one side
and output to the other side:
\begin{itemize}
  \item Conversion of an integer to a field element
  \item Zero-testing of any field element
\end{itemize}

Note that zero-testing and subtraction obviously allow equality
comparisons between field elements, but not ordering them.

We restrict the inputs and outputs to be integers (or output of $\bot$
on error). This is necessary for our use case where the inputs are
actually integers which are then reduced modulo $p$, but in any case is
not a restriction due to the integer-to-field-element conversion
instruction.

For computations over prime fields, we can formally define an
\emph{algorithm} in our algebraic RAM model as having:
\begin{itemize}
  \item Parameters: prime $p$ and initial Modular PRNG state $s$
  \item Input: $\inx \in \ZZ^n$
  \item Output: $y \in \ZZ \bigcup \{\bot\}$
  \item Program: A series of integer instructions as in a normal RAM model and
  arithmetic instructions as defined above, along with labels,
  conditional branches, and memory load/store operations
\end{itemize}

We write $\origalg_{p,s}(\inx)$ for the output $y$ of algorithm
\origalg{} with field \GF{p} and initial PRNG state $s$.

Note that some abuses of the RAM model are possible, particularly for
algebraic algorithms, by using arbitrarily large integers. Rather than
avoid such abuses by, e.g., using a word RAM model,
we merely note that our transformations do not affect the
integer side operations at all, and thus would in principle work the
same under any such model restrictions.
Indeed, our work applies over
most algebraic computational models we know of, such as straight-line
programs, branching programs, multi-tape Turing machines, or the
closely-related BSS model.

\section{Algorithm transformation for random prime fields}\label{sec:trans}

Recall the general idea of our approach, building on \citep{DDDD85,HL20}:
Given an algorithm \origalg{} which works over any prime field, we
transform it into a transformed version \transalg{} that instead samples
a possibly-composite modulus $m$ which \emph{probably} contains a large
prime factor $p\ge \sqrt{m}$.
Any time a comparison or multiplicative inverse occurs, \transalg{} first
performs a GCD of the operand with the current modulus $m$. If a
nontrivial factor $k$ of $m$ is found, then the (unknown) large prime
factor of $m$ must divide $\max(k, m/k)$. Update the modulus accordingly
and continue.

We first fully present our algorithmic transformation, then examine some
thorny issues related to the use of pseudorandom numbers in the
algorithms, and finally prove the correctness and performance bounds
which are the main results of this paper.

\subsection{Transformation procedure}\label{ssec:trans}

We begin with the crucial subroutine
\cref{alg:d5mod} (\NewModulus), which shows how to update the modulus while
carefully ensuring no significant blow-up in bit or arithmetic
complexity.

\begin{algorithm}[htbp]
  \caption{\NewModulus($a,m$)\label{alg:d5mod}}
  \KwIn{Integer $a\in\ZZ$ and modulus $m\in\NN$}
  \KwOut{New modulus $m'\in\NN$}
  \medskip

  $g_1 \gets \gcd(a,m)$ \;
  \If{$g_1^2 > m$}{\KwRet{$g_1$}}
  \Else{
    $b \gets g_1^{\floor{\log_2 m}} \bmod m$ \;
    $g_2 \gets \gcd(b, m)$ \;
    \KwRet{$m/g_2$}
  }
\end{algorithm}

\begin{lemma}\label{lem:d5corr}
  Given any $a,m\in\ZZ$, the integer $m'\in\NN$ returned by
  \NewModulus($a,m$) has the following properties:
  \begin{itemize}[nosep]
    \item $m' \divides m$
    \item $a$ is either zero or invertible modulo $m'$
    \item If $p$ is a prime with $p^2>m$ and $p\divides m$, then
    $p\divides m'$ also
  \end{itemize}
\end{lemma}
\begin{proof}
  Let $a,m\in\ZZ$ be arbitrary and $g_1, b, g_2$ as in
  the algorithm.

  For the first property, we see that the integer returned is either
  $g_1$ or $m/g_2$, both of which are always divisors of $m$.

  For the second property, consider two cases.
  First, if $g_1^2 > m$ and thus $m'=g_1$ is returned, then $g_1\divides
  a$, so $a$ is zero modulo $m'$.

  Otherwise, let $q\ge 2$ be any common factor of $a$ and $m$
  (if one exists).
  By the definition of gcd, $q\divides g_1$. Now let
  $k\ge 1$ such that $q^k$ is the largest power of $q$ that
  divides $m$. Because $q\ge 2$, $k \le \log_2 m$, and therefore
  $q^k \divides g_1^{\floor{\log_2 m}}$. From the gcd algorithm,
  this means that $q^k \divides g_2$, and thus $q$ does \emph{not}
  divide $m/g_2$. Hence $a$ and $m/g_2$ do not share any
  common factor, i.e., $\gcd(a,m')=1$ as required.

  For the third property, assume $m$ has a prime factor $p$ with
  $p^2>m$.
  If $p\divides g_1$, then $g_1^2>m$ and
  the new modulus $m'=g_1$ returned in the first case
  is be divisible by $p$. Otherwise, if $p$ does not divide $g_1$, then
  $p$ does not divide $g_2$ either, and hence $p\divides m/g_2$.
\end{proof}

\begin{lemma}\label{lem:d5rt}
  For any $a\in\ZZ$ and $m\in\NN$,
  the worst-case bit complexity of Algorithm \NewModulus($a,m$)
  is $\softoh{\log a + \log m}$.
\end{lemma}
\begin{proof}
  The bit complexity is dominated by the two GCD computations
  and the modular exponentiation to compute $b$. Using the
  asymptotically fast Half-GCD algorithm
  \citep{Sch71,SZ04}, and binary powering for the modular
  exponentiation, all three steps have bit cost
  within the stated bound.
\end{proof}

We now proceed to incorporate the \NewModulus{} subroutine into the
algorithm transformation procedure. Algorithm \origalg{}
will be in the model of an algebraic RAM from \cref{ssec:aram}.
We also require a companion
procedure \bounds{} which takes any input
$\inx\in\ZZ^n$ for \origalg{} and \emph{deterministically} produces a
positive integer $b$, which should be a minimal bit-length of primes to
ensure that \origalg{} produces correct results on input $\inx$ with
high probability. (In most applications we can imagine, \bounds{} is
a simple function of the input sizes and bit-lengths.)

\Cref{alg:trans} details the construction of the transformed algorithm
\transalg{} based on \origalg{} and \bounds{}.

\begin{algorithm}[htbp]
  \caption{\transalg{} produced from \origalg{} and \bounds{}\label{alg:trans}}
  \KwIn{Input $\inx\in\ZZ^n$ and PRNG initial state $s$}
  \KwOut{$y \in \ZZ\bigcup\{\bot\}$}

  $b \gets \bounds(\inx)$ \;
  $m \gets$ pseudorandom $(2b)$-bit integer stored in memory\;
  Proceed with identical instructions of \origalg{}, except:
  \begin{itemize}
    \item All conversions from integers to algebraic values are replaced
      by explicit reduction modulo $m$.
    \item All additions, subtractions, and multiplications on the
      ``algebraic side'' are replaced by integer arithmetic followed by
      explicit reduction modulo $m$.
    \item Any random field element generation instructions are replaced
      by sampling \linebreak a number in the range $[0,m)$ using the modular PRNG.
    \item All zero tests and multiplicative inverses on the ``algebraic side'' for some \linebreak
      value $a$ first call $\NewModulus(a,m)$, update the modulus
      accordingly, 
		  and then perform either a divisible-by-$m$ test or
      modular inverse computation, respectively.
  \end{itemize}
\end{algorithm}

Observe that the transformed algorithm \transalg{} no longer works in
the algebraic RAM model, but handles explicit integers only. More
precisely, each algebraic operation in \origalg{} is replaced with a
constant number of arithmetic or gcd computations on integers with
bit-length at most $2b$.

\subsection{Correlated PRNGs}\label{ssec:corr-prng}

A technical detail of our analysis requires a tight relationship between
the random choices made by \origalg{} and \transalg{} for the same input
and seed value. Even using the same PRNG for both algorithms, there is
no reason to think that a random sample in \origalg{} modulo a prime $p$
would have any relationship to a corresponding random sample in
\transalg{} modulo a multiple $m$ of $p$.

We achieve this by defining a pair of modular PRNGs --- one for
\origalg{} and one for \transalg{} --- both based on the same underlying
high-quality PRNG, and having the desired correlation property.
The constructions are based on three key insights. First, in defining
the transformed algorithm \transalg{} above, we are careful to sample
the initial modulus $m$ \emph{inside} the algorithm, rather than taking
it as input. Second, the PRNG for the algebraic algorithm \origalg{}
also samples a random modulus $m$ even though this is not directly used
in the computations, in order to match random outputs between \origalg{}
and \transalg{}. Third, our correlated PRNG construction doubles the state
size so that this modulus $m$ is chosen completely independently from
the future sampled random values; this allows us later to translate
probabilistic correctness modulo $p$ in \origalg{} to probabilistic correctness
modulo $m$ in \transalg{}.

Before describing the PRNGs in detail, it is important to stress that \emph{this is
purely a proof technique}. While these PRNG constructions are efficient
and realizable, their need is motivated only by the probabilistic
analysis that follows; in reality, we could just use any normal
high-quality PRNG and achieve the same results except for some
pathological algorithms would not arise
in practice.

For what follows, we assume the existence of a high-quality modular PRNG $G$
as defined in \cref{ssec:prng}.

We begin by describing the PRNG for the transformed algorithm
\transalg{} of \cref{alg:trans}, which we call $\overline{G'}$.

We will use two simultaneous instances of the underlying modular PRNG
$G$, and therefore double the state size so that the seed for $\overline{G'}$
can be written $s=(s_0,s_1)$. The
first instance with initial state $s_0$ is used only to generate the
initial modulus $m$ on the first step.

Afterwards, a random sample modulo some integer $t\in\NN$ is generated as
follows. If $t\divides m$, then $\overline{G'}$ first generates a random integer
modulo $m$ using $G$ with the second part of the current state $s_1$, and then
reduces the result again modulo $t$ before returning it. Because
$t\divides m$, this is indistinguishable from (though less efficient
than) randomly choosing an integer modulo $t$ directly.

Otherwise, if $t$ is not a divisor of the original modulus $m$, then
$\overline{G'}$ simply calls $G$ directly with the current state $s_1$ and
the requested modulus.

The PRNG $G'$ for the algebraic algorithm \origalg{} is almost
identical to $\overline{G'}$, except that the value of $m$ is only used \emph{inside the
PRNG}. Indeed, for a random seed and run of the algorithm modulo some
prime $p$, it is unlikely that $p\divides m$, and then this construction
doesn't change the results at all compared to using $G$ with the second
part of the state $s_1$ alone.

The need for this strange PRNG construction is the correlation that
exists between \origalg{} and \transalg{} whenever the seed
$s=(s_0,s_1)$ is the same for both, as captured in the following lemma:

\begin{lemma}\label{lem:prng-equiv}
  Let $s=(s_0,s_1)$ be a PRNG state for $G'$ or $\overline{G'}$, $m$ be the initial $2b$-bit random
  value chosen using $s_0$, and $p,m'\in\NN$ such that $p\divides m'$
  and $m'\divides m$.
  Writing $r$ as the pseudorandom result from $G'$ with current state
  $s$ modulo $p$, and $\overline{r}$ as the result from $\overline{G'}$ with the
  same state $s$ and larger modulus $m'$, then we have
  $r\equiv \overline{r} \bmod p$.
\end{lemma}
The proof follows directly from the PRNG definitions above and the
divisibility conditions given.

Besides this correlation property, we also need to know that using
this constructed
PRNG in the algebraic algorithm \origalg{} does not affect the
probabilistic correctness. Note that the assumption on the underlying
PRNG $G$ is an idealistic one, as any PRNG with fixed seed length cannot
actually produce uniformly random values in an arbitrary range.

\begin{lemma}\label{lem:prng-quality}
  If the underlying PRNG $G$ produces uniformly random values for any
  sequence of moduli,
  then for \emph{any} fixed value of $s_0$, the constructed modular PRNG
  $G'$ also produces uniformly random values for any sequence of moduli.
\end{lemma}
\begin{proof}
  Let $s_0$ be a fixed first half of the seed.
  The value of $s_0$ purely determines what $m$ is chosen, so let
  $m\in\NN$ be that arbitrary value.

  Now consider any modulus $t$ given to $G'$. If $t$ does not divide
  $m$, then $G'$ returns a value from the underlying PRNG $G$, which by
  assumption is uniformly random.

  Otherwise, if $t\divides m$, then $G$ is first sampled for some random
  value $r$ in $[0,m)$. By assumption $r$ takes on any value in this
  range with probability $\frac{1}{m}$. Then because $t$ is a divisor of
  $m$, reducing $r \bmod t$ produces each value in $[0,t)$ with
  probability $\frac{1}{t}$.
\end{proof}

From this, we can draw a crucial conclusion on the probabilistic
correctness of \origalg{} when only $s_1$ is varied.

\begin{corollary}\label{cor:prng-correct}
  For some prime $p$ and input $\inx$, suppose \origalg{}
  produces a given output $y$ with probability $1-\epsilon$ when provided an
  ideal perfectly-random number generator. Then if the underlying PRNG
  $G$ is uniformly random and for any fixed value of $s_0$, \transalg{}
  returns $y$ with the same probability $1-\epsilon$, where the probability is
  over all choices of $s_1$ only.
\end{corollary}

\subsection{Analysis}\label{ssec:analysis}

We now proceed to the main result of our paper: for any algebraic
algorithm \origalg{} that produces correct results with high probability
for primes of bit-length at least that given by \bounds{},
\cref{alg:trans} produces a randomized algorithm \transalg{} which, with
constant probability, produces the same correct results with the same
number of steps. That is, we can achieve (probabilistically) the same
results as computing modulo random primes, without actually needing to
ever conduct a primality test.

To prove this probabilistic near-equivalence, first define a \emph{run}
of an algorithm as the sequence of internal memory states for given
inputs and seed value (and in the case of an algebraic algorithm, choice
of field).

We first show that any run of the transformed algorithm \transalg{} is
equivalent to a run of the original \origalg{} with the same input and
seed for some choice of $p$. Here and for the remainder of this section,
we assume that \origalg{} and \transalg{} use the constructed PRNGs $G'$
and $\overline{G'}$ as defined in \cref{ssec:corr-prng}.

\begin{lemma}\label{lem:runequiv}
  For any input $\inx$ and seed $s$, consider the resulting run of
  \transalg{}. If $m'$ is the final stored value of $m$ in this run,
  then for any prime factor $p$ of $m'$, an identical run of
  \origalg{} is produced over \GF{p} with the same input \inx{} and seed
  $s$, where all algebraic values from the run of \transalg{} are
  reduced modulo $p$.
\end{lemma}
\begin{proof}
  Consider the memory states at some point in the program where they are
  equivalent between the two runs. We show that, no matter the next
  instruction in \origalg{}, the memory states after that instruction
  (and the corresponding instruction(s) in \transalg{} according to
  \cref{alg:trans}) are the still equivalent.

  Any arithmetic-side operations are unchanged in \cref{alg:trans}.

  Because reduction modulo $p$ is a homomorphism, any algebraic
  additions, subtractions, or multiplications also maintain equivalence.

  Let $m$ be the original modulus chosen at the beginning of
  \transalg{}. From the first point of \cref{lem:d5corr}, we know that
  $m'\divides m$, and therefore $p\divides m$ also. The same is true for
  any intermediate value of $m$ in the run.

  This means that any conversion operation in \origalg{}, reducing an
  integer modulo $p$, will be mod-$p$ equivalent to the corresponding
  operation in \transalg{}, reducing modulo the current value of $m$.

  Considering zero-test instructions, let $a$ (resp.\ $\overline{a}$) be the value of
  some algebraic value in the run of \origalg{} (resp.\ \transalg{}).
  If $a = 0$ in \GF{p}, then $p\divides \overline{a}$. Then, because $p$
  divides every modulus value $m$ in the run, $\overline{a}$ is not
  invertible modulo $m$. Then according to the second point in
  \cref{lem:d5corr}, $\overline{a}$ is zero mod $m$, and the zero test will
  have the same result.

  By the same reasoning, any multiplicative inverse instruction will
  also be equivalent between the runs of \origalg{} and \transalg{}, or
  in the case the denominator is zero, both runs will result in $\bot$.

  Finally, by using the correlated PRNGs defined in
  \cref{ssec:corr-prng}, we can apply \cref{lem:prng-equiv} to conclude
  that any random field element generation instruction also results in
  equivalent outputs at the same step of both runs.

  This covers all possible types of instructions and completes the
  proof.
\end{proof}

If the original algorithm \origalg{} is deterministic, this equivalence
of runs is enough to prove correctness of \transalg{}. Indeed, this has
been the assumption in most prior works on the D$^5$ principle, which
also often employ simpler models of computation such as straight-line
programs or deterministic computation trees.

By contrast, we want to allow \origalg{} to be randomized.
First, we affirm that any \emph{deterministic property} of the output is
preserved after our transformation to \transalg{} in the
following lemma, which follows directly from \cref{lem:runequiv}.

\begin{lemma}\label{lem:deterministic}
  Suppose $\inx\in\ZZ^n$ is an input for \origalg{}, and
  $Y\subseteq\ZZ\bigcup\{\bot\}$ is a family of outputs,
  such that for \emph{any} prime $p$ and random seed $s$, the output of
  \origalg{} on input \inx{} is always a member of $Y$. Then the output
  of \transalg{} on input \inx{} is always a member of $Y$ as well.
\end{lemma}

The more difficult case is when \origalg{} may return incorrect values.
There are two types of causes for an incorrect result:
when the prime $p$ is
one of a small set of ``unlucky'' values, or when randomly-sampled field
elements in the algorithm are unlucky and produce and incorrect result.
Although many actual algorithms only have one of these two types of
failure, we account for both types in order to have the most general
result.

For convenience of exposition, we will capture these failure modes in
the following definition:
\begin{definition}\label{def:wrong}
  Let $\inx\in\ZZ^n$ be any input for \origalg{}, $k\in\NN$ and
  $\epsilon\in\RR$ with $0\le \epsilon < 1$. We say that \origalg{}
  is $(k,\epsilon)$-correct for input \inx{} if, for all but at most $k$ primes $p$ with
  $p\ge 2^{\bounds(\inx)}$,
  running \origalg{} on \inx{}
  produces a correct output with probability at least $1-\epsilon$.
\end{definition}

The following theorem, which is the main result of our paper,
combines the prevalence of $b$-fat integers with the run-equivalence of
\transalg{} to prove probabilistic correctness.

\begin{theorem}\label{thm:prob-correct}
  Let $\inx\in\ZZ^n$, $k\in\NN$, and $\epsilon\in\RR$ with
  $0\le \epsilon < \frac{1}{2}$, and write $b=\bounds(\inx)$.
  If \origalg{} is $(k,\epsilon)$-correct
  for input \inx{}, then the probability that \transalg{} produces the
  correct output for input \inx{} is at least
  \[\frac{1}{2} - \frac{k}{2^{b-2}} - \frac{\epsilon}{2}.\]
\end{theorem}
\begin{proof}
  Let $m$ be the initial prime modulus chosen uniformly
  in the range $[2^{2b-1},2^{2b})$ by \transalg{}. Two things can make $m$ an unlucky
  choice: if it has no large prime divisor (i.e., if it is not $b$-fat),
  or if its largest prime divisor is one of the $k$ unlucky primes that
  cause \origalg{} to fail.

  \Cref{lem:bfat} tells us that the probability of the former is at most
  $\frac{1}{2}$. And, disjointly, the probability that $m$ \emph{does} have a
  large prime factor but it is one of the $k$ unlucky choices for
  \origalg{} is at most
  $k/2^{b-1}$.

  Therefore, over all choices of the first part of the PRNG initial
  state $s_0$, at least $1/2 - k/2^{b-2}$ of them lead to an $m$ with a
  prime factor $p\ge 2^b$ which is not one of the $k$ ``unlucky'' primes
  for \origalg{} on this input.

  For such ``lucky'' choices of $s_0$ and thereby $m$,
  because $p^2>m$, from \cref{lem:d5corr} we know that $p$ will
  always divide the updated modulus $m$ after any call to \NewModulus{}, and in
  particular, $p$ will divide the final modulus $m'$.
  We can therefore apply \cref{lem:runequiv} with the same $p$ for
  \emph{all} possible runs of \transalg{} with the same $s_0$.

  Finally, considering the remaining part of the PRNG initial state
  $s_1$, \cref{cor:prng-correct} tells us that in these cases
  \transalg{} produces the correct result with probability at least
  $1-\epsilon$, conditional on the previously-derived chance that $m$ is
  ``lucky''.
\end{proof}

When combined with an efficient verification algorithm, \cref{thm:prob-correct}
immediately yields a Las Vegas randomized algorithm.

But without an efficient verifier, the result seems to be not very
useful: it proves that \transalg{} is a Monte Carlo
randomized algorithm with success probability strictly less than
one-half.

Still, we can combine \cref{thm:prob-correct} with the preceding
\cref{lem:deterministic} in the case of algorithms \origalg{} which have
\emph{one-sided error}, meaning that, for any prime $p$ and initial PRNG
state $s$, the output $y$ from \origalg{} is never larger (or,
equivalently, never smaller) than the correct answer.

\begin{corollary}\label{cor:one-sided}
  Let $\mu>0$ such that,
  for any input \inx{}, algorithm \origalg{} is $(k,\epsilon)$-correct
  where
  $1 - \epsilon - k/2^{\bounds(\inx)-1} < \mu$.
  If \origalg{} furthermore has only one-sided error, then the correct
  output can be determined with high probability after
  $O(\frac{1}{\mu})$ runs of \transalg{}.
\end{corollary}
\begin{proof}
  Because of one-sided error, you can repeatedly run \transalg{} and
  take the maximum (resp.\ minimum) result if the output of \origalg{}
  never larger (resp.\ smaller) than the correct output.
\end{proof}

\section{Computing the sparsity of integer polynomials}\label{sec:sparsity}

Sparse polynomial interpolation is an important and well-studied problem in computer algebra: Given an
unknown polynomial $f\in \R[x_1,\dots,x_n]$ through a blackbox or a \emph{Straight Line Program} (SLP), one wants to recover the non-zero
coefficients of $f$ and their corresponding exponents. Of course, the main goal is to have an algorithm with a complexity that is
quasi-linear in the bit-length of the output, that is $\softoh{n t (\log d + \log h)}$, where $t=\#f$ is the number of non-zero terms, $d$ the maximal
degree and $h$ the height, i.e., largest absolute value of any coefficient.

The fastest algorithm for this task are of two kinds, depending of the model in which the polynomial is given. For Straight Line Programs, following the deterministic polynomial time algorithm of \citet{gs09}, many improvements have been made through
randomization to reach a quasi-linear complexity in every parameter of $f$, i.e., $t$, $\log d$ and $\log h$ \cite{Huang:2019}. For a polynomial given instead by a \emph{black box} for its evaluation,
following the seminal papers of \citet{BenorTiwari:1988} and \citet{KaltofenLakshman:1988},
we have now reached a quasi-linear complexity in the sparsity of $f$, see Arnold's PhD thesis \citep{arn16}. Note that an
algorithm with quasi-linear complexity in \emph{all} the parameters of $f$ in this model is still not available.

One of the main ingredients of these algorithms is that they require bounds for every parameter $(t,d,h)$ of $f$.  A few
works considered to replace these bounds with explicit randomized algorithms.  As outlined in \cite{KaltofenTrager1990,arn16}, one
can calculate such a degree bound in polynomial time but this might dominate the cost of the interpolation. The situation is
clearly different for the sparsity parameter as interpolation with early termination exists \cite{KL03}. This is only the case
with Prony-style interpolation, popularized by \citet{BenorTiwari:1988} for polynomials given as a blackbox.

In this work we will consider the model of \emph{Modular Blackbox} (MBB) that allows to control the size of the evaluation of the
polynomial. In particular, this is of great interest to efficiently deal with sparse polynomials over the integers as one
evaluation might be exponentially large than the polynomial itself.

Kronecker substitution \cite{Kronecker1882} is a fairly classical tool to reduce multivariate problems to univariate ones.  It is
easy to see that this transformation does not change the size of the polynomial and its sparsity. Therefore, we will only focus on
the univariate case here for simplicity of presentation.

\def\getSparsity{\textsc{GetSparsity}}
\subsection{Sparsity over a sufficiently large field}

First we develop a Monte Carlo algorithm to compute $\#f$ over a sufficiently-large
finite field \GF{q}.
For this, we can use \citet{BenorTiwari:1988}
and the extension of \citet{KL03} that study the probability that some early zeros appear
during the course of the Berlekamp-Massey algorithm.

Taking a random $\alpha\in \GF{q}$ and $a_i=f(\alpha^i) \in \GF{q}$, it is shown that for $s=1,\dots,t$ all Hankel
matrices $H_s=[a_{i+j}]_{i,j=0}^{s-1}$ are non-singular with a probability greater than $1-\frac{D t (t-1)(t+1)}{3q}$; see
\citep{KL03,Grigorescu:2010,arn16}. The so-called early termination strategy for sparse interpolation is then to run
the Berlekamp-Massey algorithm on the infinite sequence $(a_0,a_1,a_2,\dots)$ and to stop the algorithm whenever a zero discrepancy
occurs. It is showed in \citep{KL03} that the zero discrepancy corresponds exactly to hitting a singular Hankel matrix $H_s$.
Since the sequence $(a_1,a_2, \dots)$ corresponds to the evaluation of a $t$-sparse polynomial at a geometric sequence, the
minimal generator $\Lambda$ of the recurrence sequence $(a_0,a_1,a_2,\dots)$ has degree exactly $\#f$
\cite{BenorTiwari:1988}.

\begin{fact}\label{lem:getsparsity-moq}
  Given a blackbox for $f\in\GF{q}[X]$ with $q\ge 16D^4$ where $D>\deg f$, there exists a Monte Carlo algorithm 
  that computes an integer $t$ such that $t\le\#f$. With probability at least $1-\frac{1}{48}$,
  we have $t=\#f$ exactly.  The computation requires $2t$ probes
  to the modular blackbox and $\softoh{t}$ arithmetic operations in $\GF{q}$.
\end{fact}

We note that there is nothing special about the constant 16; this just arises from what we need
later, and it is convenient to have a very low probability of error.

If $f$ were given by a straight-line program instead of a blackbox, the same algorithm may
be employed with bit complexity $\softoh{L\#f\log q}$, where $L$ is the length of the SLP.

Correctness comes from the previous discussion on the probability that the Hankel matrices $H_s$ are non-singular up to $s=t$;
see \cite[Theorem 9]{KL03} for a complete proof.  For the complexity, we can use the fast iterative order basis algorithm of
\cite{GioLeb14} since the Berlekamp-Massey algorithm is related to Padé approximant involving the series $\sum_{i>0} a_ix^i$
\cite{Dornstetter87}.  The algorithm \textsf{iPM-basis} from \cite{GioLeb14} provides a fast iterative variant for Padé
approximation that can incorporate the early termination strategy (looking for a zero constant term in the residual, denoted
$F^{v}$).

One may remark that
the algorithm of \cref{lem:getsparsity-moq}
is a \emph{one-sided} randomized algorithm; the returned value $t$ never exceeds the true sparsity
$\#f$.

\subsection{Sparsity over the integers}

Now suppose $f\in\ZZ[x]$ is an integer polynomial given via a modular blackbox,
along with bounds $D,H$ such that $\deg f < D$ and each coefficient of $f$ is bounded by
$H$ in absolute value.
We want to use the techniques of \cref{sec:trans} to adapt the algorithm of \cref{lem:getsparsity-moq}
to find the sparsity of $f$.

The first question is how to incorporate the modular blackbox into the algebraic RAM
model of \cref{ssec:aram}. We will say that the algebraic RAM is endowed with an additional
instruction to probe the MBB: given any \emph{algebraic} value $\alpha$, the MBB instruction
returns a new algebraic value for $f(\alpha)$. In the original algorithm \origalg{}, each MBB
evaluation will be modulo $p$, and in the transformed algorithm \transalg{}, evaluations will
be modulo the current value of $m$.

Observe that this functionality is exactly what is already specified in the definition of a modular
black box. Importantly, we do \emph{not} require the MBB to be given in any particular computational
model (such as algebraic RAM), and the instruction transformations described in \cref{alg:trans} will
\emph{not} apply inside the blackbox itself.

The next question is how large the prime $p$ should be to ensure correctness with high
probability. \Cref{lem:getsparsity-moq} gives a lower bound for $p$ so that the
mod-$p$ algorithm succeeds, but we also need to ensure that the sparsity of $f$ modulo $p$
is the same as the actual value of $\#f$ over the integers. This will be true as long as none
of the coefficients of $f$ vanish modulo $p$. Then we simply observe that, with bounds $D,H$ on
the degree and height of $f$ respectively, the number of ``bad primes''
which cause the sparsity to drop modulo $p$ is at most $D\log_2 H$.

Set $b=\lceil 4 + 4\log_2 D + \log_2\log_2 H\rceil$. Then any $p\ge 2^b$ satisfies
the condition of \cref{lem:getsparsity-moq}, so we can say the algorithm is
$(D\log_2 H, \frac{1}{48})$-correct by \cref{def:wrong}. The following theorem, our main
result for this section, follows immediately after observing that $D\log_2 H / 2^{b-2} < \frac{1}{4}$.

\begin{theorem}\label{thm:getsparsity-z}
  Given a MBB for $f\in\ZZ[x]$ and bounds $D,H$ with $\deg f < D$ and each coefficient of $f$
  is at most $H$ in absolute value, there exists a Monte Carlo randomized algorithm that computes
  an integer $t$ such that $t\le \#f$. With probability at least $0.239$, we have
  $t=\#f$ exactly. The computation requires $2t$ probes to the MBB and $\softoh{t}$ arithmetic
  operations, all with moduli that have bit-length $\oh{\log D + \llog H}$.
\end{theorem}

Because the error is again one-sided, \cref{cor:one-sided} applies
and can be used to make the success probability arbitrarily high.

If the MBB for $f$ is in fact a straight-line program of length $L$, the total bit complexity becomes
$\softoh{Lt(\log D + \llog H)}$. While this is technically \emph{sub}-linear in the bit-length of
$f$ itself, we note that this is somewhat ``hiding'' some computation in the evaluation
model itself, since to actually produce a polynomial with degree $D$ and height $H$ with
bounded constants, the length $L$ of the SLP would need to be at least
$\Omega(\log D + \log H)$.

\section{Random irreducible polynomials without irreduci\-bility testing} \label{sec:poly}

In this section, we adapt our approach to computing in a field extension of a fixed finite field
$\GF q$. This need arises frequently in settings where the base field $\GF q$ is too small, and one needs
to find more than $q$ distinct elements in it. In that case, the standard approach is to compute a
random irreducible polynomial $\varphi$ of degree $s$ and to work within $\GF{q^s}=\GF{q}[x]/\langle
\varphi\rangle$. If the algorithm that is run in $\GF{q^s}$ is fairly fast, the cost of producing an
irreducible polynomial of degree $s$, $\softoh{s^3\log q}$, may become predominant.

We show how to adapt our techniques to compute modulo an arbitrary random polynomial. Many aspects
are very similar to the integer case, so we highlight only the main differences. We begin with the polynomial counterpart of
the notion of $b$-fat integers given in \cref{sec:primes}. 

\begin{definition}
  For any positive integer $d$ and finite field $\GF q$, 
  a polynomial $f\in\GF q[x]$ is said to be $d$-fat if
  it has an irreducible factor of degree $> d$.
\end{definition}

\begin{lemma}\label{lem:dfat}
  For any $d$, the number of degree-$2d$ monic polynomials over $\GF q$ that are $d$-fat is
  at least $\frac{1}{4}q^{2d}$.
\end{lemma}

\begin{proof}
  A monic degree-$2d$ polynomial can have at most one monic irreducible factor of degree $>d$.  A
  given monic irreducible polynomial $g$ of degree $\ell > d$ divides exactly $q^{2d-\ell}$ monic
  polynomials of degree $2d$. Moreover, the number of irreducible monic polynomials of degree $\ell$
  over $\GF q$ is at least $q^\ell/2\ell$ for any $\ell$~\cite[Lemma~19.12]{Shoup:2008}. 
  
  Therefore, there are at least $q^{2d}/2\ell$ monic degree-$2d$ polynomials of $\GF q[x]$ that have
  an irreducible factor of degree $\ell$, for $\ell > d$.  Summing from $\ell = d+1$ to $2d$, there
  are at least $\frac{1}{2}q^{2d} (H_{2d}-H_d)$ where $H_n = \sum_{i=1}^n 1/i$ denotes the $i$th
  harmonic number. Since $1/2(n+1)<H_n-\ln(n)-\gamma <1/2n$ \cite{Young:1991}, $H_{2d}-H_d \ge
  \ln(2) + 1/2(2d+1) - 1/2d \ge \frac{1}{2}$ for $d\ge 2$, and $H_2-H_1 = \frac{1}{2}$. The result
  follows.
\end{proof}

Note that the bound of that lemma is smaller than in the integer version, since we only proved that
at least one fourth of the degree-$2d$ polynomials over $\GF q$ are $d$-fat. Actually, the
number of monic irreducible polynomials of degree $\ell$ approaches $q^\ell/\ell$ for large values
of $q$ or $\ell$~\cite{Shoup:2008}. With the same proof, this shows that the fraction of degree-$2d$
polynomials that are $d$-fat approaches $\ln(2)\ge0.693$ for large values of $d$ or $q$.

We now turn to the algorithm transformation. We work in the same algebraic RAM model. We adapt the
definition of an \emph{algorithm} of \cref{sec:model}. The parameter $p$ is replaced by a parameter
$\varphi$ which is an irreducible degree-$s$ polynomial over $\GF q$. This requires to fix a
correspondence between integers and polynomials over $\GF q$. This is easily done by using the
$q$-adic expansion of integers. 

The \NewModulus{} algorithm works similarly \emph{mutatis mutandis}. As input, the algorithm takes a
polynomial over $\GF q$ (or equivalently an integer that represents this polynomial) and a modulus
$m\in\GF q[x]$, and returns a new modulus $m'\in\GF q[x]$. The test ``$g_1^2>m$'' is replaced by a test
``$2\deg(g_1) > \deg(m)$'', and ``$g_1^{\lfloor\log_2 m\rfloor}\bmod m$'' is replaced by
``$g_1^{\deg m}\bmod m$''. The proof of \cref{lem:d5corr} is easily adapted. The worst-case bit
complexity of the adapted algorithm is $\softoh{(\deg a+\deg m)\log q}$.

Finally, the algorithm transformation itself is easily adapted. The bound $\bounds(\inx)$ returns the
minimal degree $s$ of an extension for the algorithm to produce correct results with high
probability.  The transformed algorithm \transalg{} computes $2s$ pseudorandom integers modulo $q$
using the modular PRNG, and interprets them as a monic degree-$2s$ polynomial $m$ over $\GF q$.
Also, conversions from integers to algebraic values are done by interpreting the integer as a
polynomial over $\GF q$ and reducing it modulo $m$. 

The rest of the arguments are similar as in the integer case. Using the bound of \cref{lem:dfat} on
the density of $d$-fat polynomials, we obtain a similar theorem as \cref{thm:prob-correct}: If
\origalg{} is $(k,\epsilon)$-correct of input $\inx$ and $\bounds(\inx)$ returns $s$, the
transformed algorithm \transalg{} produces the correct output for input $\inx$ with probability at
least 
\[\frac{1}{4} - \frac{k}{q^{s+1}} - \frac{\epsilon}{4}.\]
Again, this error probability means the technique is only useful for one-sided Monte Carlo algorithms, or combined with an
efficient verification algorithm.

\newcommand{\Gathen}{\relax}\newcommand{\Hoeven}{\relax}


\begin{thebibliography}{38}
\providecommand{\natexlab}[1]{#1}
\providecommand{\url}[1]{\texttt{#1}}
\expandafter\ifx\csname urlstyle\endcsname\relax
  \providecommand{\doi}[1]{doi: #1}\else
  \providecommand{\doi}{doi: \begingroup \urlstyle{rm}\Url}\fi

\bibitem[Arnold(2016)]{arn16}
Andrew Arnold.
\newblock \emph{Sparse Polynomial Interpolation and Testing}.
\newblock PhD thesis, University of Waterloo, 2016.
\newblock URL \url{http://hdl.handle.net/10012/10307}.

\bibitem[Arnold et~al.(2015)Arnold, Giesbrecht, and Roche]{AGR15}
Andrew Arnold, Mark Giesbrecht, and Daniel~S. Roche.
\newblock Faster sparse multivariate polynomial interpolation of straight-line
  programs.
\newblock \emph{Journal of Symbolic Computation}, 2015.
\newblock ISSN 0747-7171.
\newblock \doi{10.1016/j.jsc.2015.11.005}.

\bibitem[Ben-Or and Tiwari(1988)]{BenorTiwari:1988}
Michael Ben-Or and Prasoon Tiwari.
\newblock A deterministic algorithm for sparse multivariate polynomial
  interpolation.
\newblock In \emph{Proceedings of the Twentieth Annual ACM Symposium on Theory
  of Computing}, STOC '88, page 301–309, New York, NY, USA, 1988. Association
  for Computing Machinery.
\newblock \doi{10.1145/62212.62241}.

\bibitem[Brent et~al.(1980)Brent, Gustavson, and Yun]{BGY80}
Richard~P. Brent, Fred~G. Gustavson, and David Y.~Y. Yun.
\newblock Fast solution of toeplitz systems of equations and computation of
  pad{\'{e}} approximants.
\newblock \emph{J. Algorithms}, 1\penalty0 (3):\penalty0 259--295, 1980.
\newblock \doi{10.1016/0196-6774(80)90013-9}.

\bibitem[Dahan et~al.(2006)Dahan, Moreno~Maza, Schost, and Xie]{DMSX06}
Xavier Dahan, Marc Moreno~Maza, {\'E}ric Schost, and Yuzhen Xie.
\newblock On the complexity of the d5 principle.
\newblock In \emph{Transgressive Computing 2006: A conference in honor of
  {J}ean {D}ella {D}ora}, pages 149--168, 2006.

\bibitem[Della~Dora et~al.(1985)Della~Dora, Dicrescenzo, and Duval]{DDDD85}
Jean Della~Dora, Claire Dicrescenzo, and Dominique Duval.
\newblock About a new method for computing in algebraic number fields.
\newblock In Bob~F. Caviness, editor, \emph{EUROCAL '85: European Conference on
  Computer Algebra Linz, Austria, April 1--3 1985 Proceedings Vol. 2: Research
  Contributions}, pages 289--290, Berlin, Heidelberg, 1985. Springer Berlin
  Heidelberg.
\newblock ISBN 978-3-540-39685-7.
\newblock \doi{10.1007/3-540-15984-3_279}.

\bibitem[Diaz-Toca and Lombardi(2010)]{DL10}
Gemma~Maria Diaz-Toca and Henri Lombardi.
\newblock Dynamic galois theory.
\newblock \emph{Journal of Symbolic Computation}, 45\penalty0 (12):\penalty0
  1316--1329, 2010.
\newblock \doi{10.1016/j.jsc.2010.06.012}.
\newblock MEGA’2009.

\bibitem[Dornstetter(1987)]{Dornstetter87}
Jean~Louis Dornstetter.
\newblock On the equivalence between {B}erlekamp's and {E}uclid's algorithms.
\newblock \emph{IEEE Transactions on Information Theory}, 33\penalty0
  (3):\penalty0 428--431, 1987.
\newblock \doi{10.1109/TIT.1987.1057299}.

\bibitem[Duval(1989)]{duval89}
Dominique Duval.
\newblock Rational {Puiseux} expansions.
\newblock \emph{Compositio Mathematica}, 70\penalty0 (2):\penalty0 119--154,
  1989.

\bibitem[Duval and Reynaud(1994)]{DR94}
Dominique Duval and Jean-Claude Reynaud.
\newblock Sketches and computation--ii: dynamic evaluation and applications.
\newblock \emph{Mathematical Structures in Computer Science}, 4\penalty0
  (2):\penalty0 239--271, 1994.

\bibitem[Garg and Schost(2009)]{gs09}
Sanchit Garg and {\'E}ric Schost.
\newblock Interpolation of polynomials given by straight-line programs.
\newblock \emph{Theoretical Computer Science}, 410\penalty0 (27-29):\penalty0
  2659--2662, 2009.
\newblock ISSN 0304-3975.
\newblock \doi{10.1016/j.tcs.2009.03.030}.

\bibitem[\Gathen{von zur Gathen} and Gerhard(2013)]{Gathen:2013}
Joachim \Gathen{von zur Gathen} and J\"{u}rgen Gerhard.
\newblock \emph{Modern Computer Algebra (third edition)}.
\newblock Cambridge University Press, 2013.
\newblock ISBN 9781107039032.

\bibitem[Giorgi and Lebreton(2014)]{GioLeb14}
Pascal Giorgi and Romain Lebreton.
\newblock Online order basis and its impact on block {W}iedemann algorithm.
\newblock In \emph{Proceedings of the 2014 international symposium on symbolic
  and algebraic computation}, ISSAC'14, pages 202--209. ACM, 2014.
\newblock \doi{10.1145/2608628.2608647}.

\bibitem[Grigorescu et~al.(2010)Grigorescu, Jung, and
  Rubinfeld]{Grigorescu:2010}
Elena Grigorescu, Kyomin Jung, and Ronitt Rubinfeld.
\newblock A local decision test for sparse polynomials.
\newblock \emph{Information Processing Letters}, 110\penalty0 (20):\penalty0
  898--901, 2010.
\newblock ISSN 0020-0190.
\newblock \doi{https://doi.org/10.1016/j.ipl.2010.07.012}.

\bibitem[\Hoeven{van der Hoeven} and Lecerf(2015)]{HL15}
Joris \Hoeven{van der Hoeven} and Gr{\'e}goire Lecerf.
\newblock Sparse polynomial interpolation in practice.
\newblock \emph{ACM Commun. Comput. Algebra}, 48\penalty0 (3/4):\penalty0
  187--191, February 2015.
\newblock \doi{10.1145/2733693.2733721}.

\bibitem[\Hoeven{van der Hoeven} and Lecerf(2020)]{HL20}
Joris \Hoeven{van der Hoeven} and Gr\'egoire Lecerf.
\newblock Directed evaluation.
\newblock \emph{Journal of Complexity}, 60, 2020.
\newblock \doi{10.1016/j.jco.2020.101498}.

\bibitem[Huang(2019)]{Huang:2019}
Qiao-Long Huang.
\newblock Sparse {{Polynomial Interpolation}} over {{Fields}} with {{Large}} or
  {{Zero Characteristic}}.
\newblock In \emph{Proceedings of the 2019 on {{International Symposium}} on
  {{Symbolic}} and {{Algebraic Computation}} - {{ISSAC}} '19}, pages 219--226,
  {Beijing, China}, 2019. {ACM Press}.
\newblock \doi{10.1145/3326229.3326250}.

\bibitem[Javadi and Monagan(2010)]{JM10}
Seyed Mohammad~Mahdi Javadi and Michael Monagan.
\newblock Parallel sparse polynomial interpolation over finite fields.
\newblock In \emph{Proceedings of the 4th International Workshop on Parallel
  and Symbolic Computation}, PASCO '10, page 160–168, New York, NY, USA,
  2010. Association for Computing Machinery.
\newblock \doi{10.1145/1837210.1837233}.

\bibitem[Joye et~al.(2000)Joye, Paillier, and Vaudenay]{JPV00}
Marc Joye, Pascal Paillier, and Serge Vaudenay.
\newblock Efficient generation of prime numbers.
\newblock In {\c{C}}etin~K. Ko{\c{c}} and Christof Paar, editors,
  \emph{Cryptographic Hardware and Embedded Systems --- CHES 2000}, pages
  340--354, Berlin, Heidelberg, 2000. Springer Berlin Heidelberg.

\bibitem[Kaltofen(2010)]{Kal10a}
Erich Kaltofen.
\newblock Fifteen years after {DSC} and {WLSS2}: {W}hat parallel computations
  {I} do today [invited lecture at {PASCO} 2010].
\newblock In \emph{Proceedings of the 4th International Workshop on Parallel
  and Symbolic Computation}, PASCO '10, pages 10--17, New York, NY, USA, 2010.
  ACM.
\newblock ISBN 978-1-4503-0067-4.
\newblock \doi{10.1145/1837210.1837213}.

\bibitem[Kaltofen and Lee(2003)]{KL03}
Erich Kaltofen and {Wen-shin} Lee.
\newblock Early termination in sparse interpolation algorithms.
\newblock \emph{Journal of Symbolic Computation}, 36\penalty0 (3-4):\penalty0
  365--400, 2003.
\newblock ISSN 0747-7171.
\newblock \doi{10.1016/S0747-7171(03)00088-9}.
\newblock ISSAC 2002.

\bibitem[Kaltofen and Trager(1990)]{KaltofenTrager1990}
Erich Kaltofen and Barry~M. Trager.
\newblock Computing with {Polynomials} {Given} {By} {Black} {Boxes} for {Their}
  {Evaluations}: {Greatest} {Common} {Divisors}, {Factorization}, {Separation}
  of {Numerators} and {Denominators}.
\newblock \emph{Journal of Symbolic Computation}, 9\penalty0 (3):\penalty0
  301--320, 1990.
\newblock \doi{10.1016/S0747-7171(08)80015-6}.

\bibitem[Kaltofen and Yagati(1988)]{KaltofenLakshman:1988}
Erich Kaltofen and Lakshman Yagati.
\newblock Improved sparse multivariate polynomial interpolation algorithms.
\newblock In P.~Gianni, editor, \emph{Proc. ISSAC}, pages 467--474, 1988.
\newblock \doi{10.1007/3-540-51084-2_44}.

\bibitem[Kaltofen et~al.(1990)Kaltofen, Lakshman, and Wiley]{kaltofen90}
Erich Kaltofen, Yagati~N. Lakshman, and John-Michael Wiley.
\newblock Modular rational sparse multivariate polynomial interpolation.
\newblock In \emph{Proceedings of the international symposium on {Symbolic} and
  algebraic computation}, {ISSAC} '90, pages 135--139, Tokyo, Japan, 1990. ACM
  Press.
\newblock \doi{10.1145/96877.96912}.

\bibitem[Karp and Rabin(1987)]{KR87}
Richard~M. Karp and Michael~O. Rabin.
\newblock Efficient randomized pattern-matching algorithms.
\newblock \emph{IBM Journal of Research and Development}, 31\penalty0
  (2):\penalty0 249--260, 1987.
\newblock \doi{10.1147/rd.312.0249}.

\bibitem[Kronecker(1882)]{Kronecker1882}
Leopold Kronecker.
\newblock Grundzüge einer arithmetischen theorie der algebraischen grössen.
\newblock \emph{Journal für die reine und angewandte Mathematik}, 92:\penalty0
  1--122, 1882.

\bibitem[Lazard(1991)]{lazard91}
Daniel Lazard.
\newblock A new method for solving algebraic systems of positive dimension.
\newblock \emph{Discrete Applied Mathematics}, 33\penalty0 (1):\penalty0
  147--160, 1991.
\newblock \doi{10.1016/0166-218X(91)90113-B}.

\bibitem[Maurer(1995)]{Mau95}
Ueli~M. Maurer.
\newblock Fast generation of prime numbers and secure public-key cryptographic
  parameters.
\newblock \emph{Journal of Cryptology}, 8\penalty0 (3):\penalty0 123--155,
  1995.
\newblock \doi{10.1007/BF00202269}.

\bibitem[Maza and Rioboo(1995)]{mr95}
Marc~Moreno Maza and Renaud Rioboo.
\newblock Polynomial gcd computations over towers of algebraic extensions.
\newblock In \emph{Applied Algebra, Algebraic Algorithms and Error-Correcting
  Codes}, pages 365--382, Berlin, Heidelberg, 1995. Springer Berlin Heidelberg.
\newblock \doi{10.1007/3-540-60114-7_28}.

\bibitem[Moenck(1973)]{Moenck73}
Robert~T. Moenck.
\newblock Fast computation of gcds.
\newblock In \emph{Proceedings of the 5th Annual {ACM} Symposium on Theory of
  Computing, April 30 - May 2, 1973, Austin, Texas, {USA}}, pages 142--151.
  {ACM}, 1973.
\newblock \doi{10.1145/800125.804045}.

\bibitem[Neiger et~al.(2021)Neiger, Salvy, Schost, and Villard]{NSSV21}
Vincent Neiger, Bruno Salvy, {\'{E}}ric Schost, and Gilles Villard.
\newblock Faster modular composition.
\newblock \emph{CoRR}, abs/2110.08354, 2021.
\newblock URL \url{https://arxiv.org/abs/2110.08354}.

\bibitem[Noro(2006)]{Noro06}
Masayuki Noro.
\newblock Modular dynamic evaluation.
\newblock In \emph{Proceedings of the 2006 International Symposium on Symbolic
  and Algebraic Computation}, ISSAC '06, page 262–268, New York, NY, USA,
  2006. Association for Computing Machinery.
\newblock \doi{10.1145/1145768.1145812}.

\bibitem[Poteaux and Weimann(2021)]{Poteaux21}
Adrien Poteaux and Martin Weimann.
\newblock Computing puiseux series: a fast divide and conquer algorithm.
\newblock \emph{Annales Henri Lebesgue}, 2021.

\bibitem[Rosser and Schoenfeld(1962)]{RS62}
J.~Barkley Rosser and Lowell Schoenfeld.
\newblock Approximate formulas for some functions of prime numbers.
\newblock \emph{Ill. J. Math.}, 6:\penalty0 64--94, 1962.
\newblock URL \url{http://projecteuclid.org/euclid.ijm/1255631807}.

\bibitem[Sch{\"o}nhage(1971)]{Sch71}
Arnold Sch{\"o}nhage.
\newblock Schnelle berechnung von kettenbruchentwicklungen.
\newblock \emph{Acta Informatica}, 1\penalty0 (2):\penalty0 139--144, Jun 1971.
\newblock \doi{10.1007/BF00289520}.

\bibitem[Shoup(2008)]{Shoup:2008}
Victor Shoup.
\newblock \emph{A {{Computational Introduction}} to {{Number Theory}} and
  {{Algebra}}}.
\newblock {Cambridge University Press}, 2008.
\newblock ISBN 978-0-521-51644-0.

\bibitem[Stehl{\'e} and Zimmermann(2004)]{SZ04}
Damien Stehl{\'e} and Paul Zimmermann.
\newblock A binary recursive gcd algorithm.
\newblock In Duncan Buell, editor, \emph{Algorithmic Number Theory: 6th
  International Symposium, ANTS-VI, Burlington, VT, USA, June 13-18, 2004,
  Proceedings}, pages 411--425, Berlin, Heidelberg, 2004. Springer Berlin
  Heidelberg.
\newblock \doi{10.1007/978-3-540-24847-7_31}.

\bibitem[Young(1991)]{Young:1991}
Robert~M. Young.
\newblock 75.9 {{Euler}}'s {{Constant}}.
\newblock \emph{The Mathematical Gazette}, 75\penalty0 (472):\penalty0
  187--190, 1991.
\newblock \doi{10.2307/3620251}.

\end{thebibliography}
\end{document}